%% file: root.tex
\documentclass[letterpaper, 10 pt, conference]{ieeeconf}  % Comment this line out if you need a4paper

\IEEEoverridecommandlockouts                              % This command is only needed if 
\usepackage{graphics}
\usepackage{amsmath}

\usepackage{amsthm}
\usepackage{amssymb}
\usepackage{mathtools}
\usepackage{bbm}
\usepackage{dsfont}
\usepackage{color}
\usepackage{xcolor}
\definecolor{myblue}{rgb}{0 0.4470 0.7410}
\definecolor{mypurple}{rgb}{0.4940 0.1840 0.5560}
\definecolor{mygreen}{rgb}{0.4660 0.6740 0.1880}
\usepackage[noend]{algpseudocode}
\usepackage[ruled, vlined, linesnumbered]{algorithm2e}
\usepackage{algorithmicx}

\newtheorem{prob}{Problem}
\newtheorem{theorem}{Theorem}
\newtheorem{assumption}{Assumption}
\newtheorem{rem}{Remark}

\newtheorem{ex}{Example}

\theoremstyle{remark}

\usepackage{pgfplots}
\pgfplotsset{compat=newest}
\pgfplotsset{plot coordinates/math parser=false}
\newlength\figureheight
\newlength\figurewidth

\newcommand\norm[1]{\left\lVert#1\right\rVert}

\title{\LARGE \bf
Simulator-Driven Deceptive Control \\ via Path Integral Approach
}

\author{Apurva Patil$^{1,*}$, Mustafa O. Karabag$^{2,*}$, Takashi Tanaka$^{3}$, Ufuk Topcu$^{3}$\thanks{$^{*}$ Indicates equal contribution. $^{1}$Walker Department of Mechanical Engineering, University of Texas at Austin, {\tt\small apurvapatil@utexas.edu}.
$^{2}$Department of Electrical and Computer Engineering, University of Texas at Austin, {\tt\small karabag@utexas.edu}.
$^{3}$Department of Aerospace Engineering and Engineering Mechanics, University of Texas at Austin, {\tt\small ttanaka@utexas.edu}, {\tt\small utopcu@utexas.edu
}.}
}

\begin{document}

\maketitle
\thispagestyle{empty}
\pagestyle{empty}

%%%%%%%%%%%%%%%%%%%%%%%%%%%%%%%%%%%%%%%%%%%%%%%%%%%%%%%%%%%%%%%%%%%%%%%%%%%%%%%%
\begin{abstract}
We consider a setting where a supervisor delegates an agent to perform a certain control task, while the agent is incentivized to deviate from the given policy to achieve its own goal. In this work, we synthesize the optimal deceptive policies for an agent who attempts to hide its deviations from the supervisor's policy. We study the deception problem in the continuous-state discrete-time stochastic dynamics setting and, using motivations from hypothesis testing theory, formulate a Kullback-Leibler control problem for the synthesis of deceptive policies. This problem can be solved using backward dynamic programming in principle, which suffers from the curse of dimensionality. However, under the assumption of deterministic state dynamics, we show that the optimal deceptive actions can be generated using path integral control. This allows the agent to numerically compute the deceptive actions via Monte Carlo simulations. Since Monte Carlo simulations can be efficiently parallelized, our approach allows the agent to generate deceptive control actions online. We show that the proposed simulation-driven control approach asymptotically converges to the optimal control distribution. 
%Furthermore, we provide a sample complexity rate analysis for a special class of tasks. We present a simulation study to validate the derived control synthesis framework.

%We consider a setting in which an agent is contracted by a supervisor to perform a certain task. The supervisor designs a policy for the agent to follow in order to complete the task. However, the agent aims to achieve a different task that might be malicious towards the supervisor. In this work, we synthesize the optimal deceptive policies for such an agent who attempts to hide his intentions from the supervisor. The problem is formulated as a KL control problem and the agent is assumed to follow a continuous-state discrete-time stochastic dynamics. This problem can be solved using backward dynamic programming which suffers from the curse of dimensionality. However, under the assumption of deterministic state dynamics, we show that the optimal deceptive policies can be numerically synthesized using Monte Carlo sampling. Since the Monte Carlo simulations can be efficiently parallelized on GPUs, our approach allows the agent to synthesize the optimal deceptive policies online. A simulation study is presented to validate the proposed control synthesis framework.

\end{abstract}

%%%%%%%%%%%%%%%%%%%%%%%%%%%%%%%%%%%%%%%%%%%%%%%%%%%%%%%%%%%%%%%%%%%%%%%%%%%%%%%%
\section{Introduction}

We consider a deception problem between a supervisor and an agent. The supervisor delegates an agent to perform a certain task and provides a reference policy to be followed in a stochastic environment. The agent, on the other hand, aims to achieve a different task and may deviate from the reference policy to accomplish its own task. The agent uses a deceptive policy to hide its deviations from the reference policy. In this work, we synthesize the optimal policies for such a deceptive agent.

We formulate the agent's deception problem using motivations from hypothesis testing theory. We assume that the supervisor aims to detect whether the agent deviated from the reference policy by observing the state-action paths of the agent. On the flip side, the agent's goal is to employ a deceptive policy that achieves the agent's task and minimizes the detection rate of the supervisor. We design the agent's deceptive policy that minimizes the Kullback-Leibler (KL) divergence from the reference policy while achieving the agent's task. The use of KL divergence is motivated by the log-likelihood ratio test, which is the most powerful detection test for any given significance level~\cite{cover1999elements}. Minimizing the KL divergence is equivalent to minimizing the expected log-likelihood ratio between distributions of the paths generated by the agent's deceptive policy and the reference policy. We also note that due to the Bratagnolle-Huber inequality~\cite{bretagnolle1978estimation}, for any statistical test employed by the supervisor, the sum of false positive and negative rates is lower bounded by a decreasing function of KL divergence between the agent's policy and the reference policy. Consequently, minimizing the KL divergence is a proxy for minimizing the detection rate of the supervisor. We represent the agent's task with a cost function and formulate the agent's objective function as a weighted sum of the cost function and the KL divergence.

%Due to the Bratagnolle-Huber inequality~\cite{bretagnolle1978estimation}, for any statistical test employed by the supervisor, the sum of false positive and negative rates is lower bounded by a function of Kullback-Leibler (KL) divergence between the distributions of the agent's deceptive policy and the reference policy. In detail, for a given false negative rate, the false positive rate increases if the KL divergence decreases.

We assume that the agent's environment follows discrete-time continuous-state dynamics. When the dynamics are linear, the supervisor's (stochastic) policies are Gaussian, and the cost functions are quadratic, minimizing a weighted sum of the cost function, and the KL divergence leads to solving a linear quadratic regulator problem. However, we consider a broader setting with potentially non-linear state dynamics, non-quadratic cost functions, and non-Gaussian reference policies. In this case, the agent's optimal deceptive policy does not necessarily admit a closed-form solution. While the agent's problem can be solved using backward dynamic programming, this approach suffers from the curse of dimensionality.

We show that, under the assumption of deterministic state dynamics, the optimal deceptive actions can be generated using path integral control without explicitly synthesizing a policy. In detail, we propose a two-step randomized algorithm for simulator-driven control for deception. At each time step, the algorithm first creates forward Monte Carlo samples of system paths under the reference policy. Then, the algorithm uses a cost-proportional weighted sampling method to generate a control input at that time step. We show that the proposed approach asymptotically converges to the optimal action distribution. Since Monte Carlo simulations can be efficiently parallelized, our approach allows the agent to generate the optimal deceptive actions online. 
%We also show that for a special class of binary cost functions, e.g., reach-avoid cost functions, the proposed method has a polynomial sample complexity rate.

The contributions of this paper are threefold:
1) The work studies a problem of deception under supervisory control for continuous-state discrete-time stochastic systems. Given a reference policy, we formalize the synthesis of an optimal deceptive policy as a KL control problem and solve it using backward dynamic programming. 
    2) For the deterministic state dynamics, we propose a path-integral-based solution methodology for simulator-driven control. We develop an algorithm based on Monte Carlo sampling to numerically compute the optimal deceptive actions online. Furthermore, we show that the proposed approach asymptotically converges to the optimal control distribution of the deceptive agent.
%    \item We show that the proposed method has a polynomial sample complexity rate for a special class of binary cost functions. 
3)
We present a numerical example to validate the derived simulator-driven control synthesis framework.

\subsection{Related Work}

Deception naturally occurs in settings where two parties with conflicting objectives coexist. The example domains for deception include robotics~\cite{shim2013taxonomy,dragan2015deceptive}, supervisory control settings~\cite{karabag2021deception,karabag2022exploiting}, warfare~\cite{lloyd2003art}, and cyber systems~\cite{wang2018cyber}.

We formulate a deception problem motivated by hypothesis testing. This problem has been studied for fully observable Markov decision processes~\cite{karabag2021deception}, partially observable Markov decision processes~\cite{karabag2022exploiting}, and hidden Markov models~\cite{keroglou2018probabilistic}. Different from \cite{karabag2021deception,karabag2022exploiting,keroglou2018probabilistic} that study discrete-state systems and directly solve an optimization problem for the synthesis of deceptive policies, we consider a nonlinear continuous-state system and provide a sampling-based solution for the synthesis of deceptive policies. In the security framework,  \cite{kung2016performance,bai2017data} study the detectability of an attacker in a stochastic control setting. Similar to our formulation, \cite{kung2016performance,bai2017data} provide a KL divergence-based optimization problem. While we consider an agent whose goal is to optimize a different cost function from the supervisor, \cite{kung2016performance,bai2017data} consider an attacker whose goal is to maximize the state estimation error of a controller.
%Goal deception is relevant to  \cite{dragan2015deceptive,masters2017deceptive,savas2022deceptive}

KL divergence objective is also used in reinforcement learning~\cite{schulman2015trust,filippi2010optimism} to improve the learning performance and in KL control frameworks~\cite{todorov2007linearly,ito2022kullback} for the efficient computation of optimal policies. In \cite{ito2022kullback}, Ito et al. studied the KL control problem for nonlinear continuous-state space systems and characterized the optimal policies. Different from \cite{ito2022kullback}, we provide a randomized control algorithm based on path integral approach that converges to the optimal policy as the number of samples increases. Path integral control is a sampling-based algorithm employed to solve nonlinear stochastic optimal control problems numerically ~\cite{kappen2005path, theodorou2010generalized, williams2016aggressive}. It allows the policy designer to compute the optimal control inputs online using Monte Carlo samples of system paths. The Monte Carlo simulations can be massively parallelized on GPUs, and thus the path integral approach is less susceptible to the curse of dimensionality \cite{williams2017model}.

\subsection{Notation}
Let $(\mathcal{X}, \mathcal{B}(\mathcal{X}))$ be a measurable space where $\mathcal{X}\subseteq\mathbb{R}^n$ is a Borel set and $\mathcal{B}(\mathcal{X})$ is a Borel $\sigma$-algebra. Suppose $(\Omega, \mathcal{F}, \mathcal{P})$ is a probability space. An $( \mathcal{F}, \mathcal{B}(\mathcal{X}))$-measurable random variable $X$ is a function $X:\Omega\rightarrow\mathcal{X}$ whose probability distribution $P_X$ is defined by 
\begin{equation*}
  P_X(B) = \mathcal{P}(X\in B)\quad \forall B \in \mathcal{B}(\mathcal{X})
\end{equation*}
$P_{X_2|X_1}(\cdot|\cdot):\mathcal{B}(\mathcal{X}_2) \times \mathcal{X}_1 \rightarrow [0,1]$ represents a stochastic kernel on $\mathcal{X}_2$ given $\mathcal{X}_1$. For simplicity, we write $P_X(dx)$ and $P_{X_2|X_1}(dx_2|x_1)$ as $P(dx)$ and $P(dx_2|x_1)$. If $P_1$ and $P_2$ are probability distributions on $(\mathcal{X}, \mathcal{B}(\mathcal{X}))$ then, the Kullback-Leibler (KL) divergence from $P_1$ to $P_2$ is defined as
\begin{equation*}\label{eq: KL divergence def}
	D(P_2\|P_1)=\int_\mathcal{X} \log\frac{dP_2}{dP_1}(x)P_2(dx)
\end{equation*}
if the Radon-Nikodym derivative $\frac{dP_2}{dP_1}$ exists, and $D(P_2\|P_1)=+\infty$ otherwise. Throughout this paper, we use the natural logarithm. Let $\mathcal{T}=\{0, 1, ... , T\}$ be the set of discrete time indices. A set of variables $\{x_0, x_1, \hdots, x_T\}$ is denoted by $x_{0:T}$ and a Cartesian product of sets $\mathcal{X}_0\times\mathcal{X}_1\times\hdots\times\mathcal{X}_T$ is denoted by $\mathcal{X}_{0:T}$. $P_{X_{0:T}}(dx_{0:T})$ denotes the joint probability distribution of random variables $X_0, X_1, \hdots, X_T$ on $\left(\mathcal{X}_{0:T}, \mathcal{B}(\mathcal{X}_{0:T})\right)$.

% \section{Preliminaries}
\section{Problem Formulation}\label{Sec: KL Control}
   We consider a setting in which a supervisor contracts an agent to perform a certain task. Suppose the agent operates in a stochastic environment and follows discrete-time continuous-state dynamics. Let the state transition law of the agent be denoted by $P(dx_{t+1}|x_t, u_t):\mathcal{B}(\mathcal{X}_{t+1}) \times \mathcal{X}_{t} \times \mathcal{U}_{t} \rightarrow [0,1]$, where the random variables $X_t\in\mathcal{X}_t$ and $U_t\in\mathcal{U}_t$ represent the state and the control input of the system at time $t\in\mathcal{T}$. $\mathcal{X}_t$ and $\mathcal{U}_t$ are assumed to be Euclidean spaces with appropriate dimensions. Suppose a supervisor provides a (possibly stochastic) reference policy $\{R_{U_t|X_t}(\cdot|x_t)\}_{t=0}^{T-1}$ to the agent and expects the agent to follow the policy to accomplish a certain task. Here, $R_{U_t|X_t}:\mathcal{B}(\mathcal{U}_t) \times \mathcal{X}_t \rightarrow [0,1]$ is a stochastic kernel on $\mathcal{U}_t$ given $\mathcal{X}_t$. The agent, on the other hand, aims to achieve a different task by minimizing the following cost function, which we henceforth call as \textit{path cost}:
\begin{equation}\label{eq:path cost}
   C_{0:T}(x_{0:T}, u_{0:T-1}) \coloneqq \sum_{t=0}^{T-1} C_t(x_t, u_t)+C_T(x_T)
\end{equation}
where $C_t(\cdot, \cdot):\mathcal{X}_t\times \mathcal{U}_t\rightarrow \mathbb{R}$ for $t\in\mathcal{T}$ and $C_T(\cdot): \mathcal{X}_T\rightarrow \mathbb{R}$ represent the stage costs and the terminal cost, respectively. In order to minimize the path cost \eqref{eq:path cost}, the agent designs its policy (possibly stochastic) $\{Q_{U_t|X_t}(\cdot|x_t)\}_{t=0}^{T-1}$ that may deviate from the reference policy $\{R_{U_t|X_t}(\cdot|x_t)\}_{t=0}^{T-1}$. The agent also attempts to be stealthy to hide its deviations from the supervisor. While the agent executes its policy $Q$, suppose the supervisor observes its state-action paths $\{x_{0:T}, u_{0:T-1}\}$, and uses a likelihood ratio test to detect whether the agent deviates from the reference policy. According to the Neyman–Pearson lemma, the likelihood-ratio test is optimal among all simple hypothesis tests for a given significance level~\cite{cover1999elements}. In other words, we consider the worst-case scenario for the agent to be detected by the supervisor. Suppose the initial state $X_0 = x_0$ of the agent is known. We denote the joint probability distribution of the state-action paths induced via the reference policy by
\begin{align}\label{eq:def_ref_traj_dist}
 R_{X_{0:T}\times U_{0:T-1}}& (dx_{0:T}\times du_{0:T-1})\nonumber\\ 
    &=\prod_{t=0}^{T-1}\!\! P(dx_{t+1}|x_t, u_t)R(du_t|x_t), 
\end{align}
and the joint distribution induced via the agent's policy by 
\begin{align}\label{eq:def_dc_traj_dist}
    Q_{X_{0:T}\times U_{0:T-1}}& (dx_{0:T}\times du_{0:T-1})\nonumber\\ 
    &=\prod_{t=0}^{T-1}\!\! P(dx_{t+1}|x_t, u_t)Q(du_t|x_t). 
\end{align}
%  The reference policy could be stochastic to provide robustness against unmodeled system components, encourage exploration, and account for input sensor noises.
% Suppose the probability of a path $\{x_{0:T}, u_{0:T-1}\}$ under the reference distribution \eqref{eq:def_ref_traj_dist} and the agent's distribution \eqref{eq:def_dc_traj_dist} are denoted by $\Pr\left(\{x_{0:T}, u_{0:T-1}\}|R\right)$ and $\Pr\left(\{x_{0:T}, u_{0:T-1}\}|Q\right)$, respectively.
Given a path $\{x_{0:T}, u_{0:T-1}\}$ that is randomly sampled under the agent's policy, the supervisor computes the log-likelihood ratio (LLR) 
\begin{equation}
    \pi(x_{0:T}, u_{0:T-1})=\log \frac{dQ_{X_{0:T}\times U_{0:T-1}}}{dR_{X_{0:T}\times U_{0:T-1}}}(x_{0:T}, u_{0:T-1}) \label{eq:LLR}.
\end{equation} The supervisor decides that the agent uses the reference policy $R$ if $\pi(x_{0:T}, u_{0:T-1})\leq c$, and deviates from $R$ otherwise. Here $c$ is a constant chosen by the supervisor to obtain a specified significance level. An agent not wanting to be detected by the supervisor must minimize the LLR \eqref{eq:LLR}. However, since the agent's trajectories are stochastic, the agent cannot directly minimize the LLR. We consequently consider that the agent's goal is to minimize the expected LLR as follows:
\begin{equation}\label{eq:ELLR}
   \Pi = \mathbb{E}_Q\left[\log \frac{dQ_{X_{0:T}\times U_{0:T-1}}}{dR_{X_{0:T}\times U_{0:T-1}}}(x_{0:T}, u_{0:T-1})\right]
\end{equation} 
where $\mathbb{E}_Q[\cdot]$ represents the expectation with respect to the probability distribution $Q$ \eqref{eq:def_dc_traj_dist}. Note that equation \eqref{eq:ELLR} also defines the Kullback-Leibler (KL) divergence $D(Q \|R)$ between the agent's distribution $Q$  and the reference distribution $R$. It can be shown that $D(Q \|R)$ can be written as the stage-additive KL divergence between $Q_{U_t|X_t}$ and $R_{U_t|X_t}$ as follows (see Appendix A):
\begin{align}\label{eq: KL cost}
   \!\! D(Q \|R)\! =\! \mathbb{E}_Q\!\left[\sum_{t=0}^{T-1} D(Q_{U_t|X_t}(\cdot | X_t) \|R_{U_t|X_t}(\cdot | X_t))\right].
\end{align}
Since the KL divergence $ D(Q \|R)$ is equivalent to the expected LLR \eqref{eq:ELLR}, in this work, the KL divergence is used as a proxy for the measure of the agent's deviations from the reference policy. 

Minimizing the KL divergence is in fact equivalent to minimizing the detection rate of an attacker for an ergodic process as proved in \cite{bai2017data}. While we do not consider an ergodic process, the use of KL divergence is still well-motivated by the Bretagnolle–Huber inequality~\cite{bretagnolle1978estimation}. Let \(\mathcal{E}\) be an arbitrary set of events that the supervisor will identify the agent as a deceptive agent, i.e., the agent followed \(Q\). According to the Bretagnolle–Huber inequality, we have
 \begin{equation}\label{eq: BH inequlity}
     \Pr(\mathcal{E}|R) + \Pr(\neg \mathcal{E}|Q) \geq \frac{1}{2} \exp(-D(Q||R))
 \end{equation}
 where $\Pr(\mathcal{E}|R)$ and $\Pr(\neg \mathcal{E}|Q)$ denote the supervisor's false positive and negative rates, respectively. The false positive rate is the probability that the supervisor will identify the well-intentioned agent as a deceptive agent, i.e., the agent's policy is \(R\), but the supervisor thinks that the agent has followed \(Q\). Similarly, the false negative rate is the probability that the supervisor will identify the deceptive agent as a well-intentioned agent. The Bretagnolle–Huber inequality \eqref{eq: BH inequlity} states that the sum of the supervisor's false positive and negative rates is lower bounded by a decreasing function of the KL divergence between the distributions \(Q\) and \(R\). 
 %the Chernoff-Stein lemma~\cite{elementsofinfortheory}. False positive rate \(\alpha\) is the probability that the supervisor will identify the well-intentioned agent as a deceptive agent, i.e., the agent's policy is \(R\), but the supervisor thinks that the agent has followed \(Q\). Similarly, false negative rate \(\beta\) is the probability that the supervisor will identify the deceptive agent as a well-intentioned agent, i.e., the agent's policy is \(Q\), but the supervisor thinks that the agent has followed \(R\). Chernoff-Stein lemma states that for a fixed level of false positive rate \(\alpha\), the best achievable false negative rate \(\beta^{*}\) decays exponentially with the number of samples $n$ and the KL divergence from \(R\) to \(Q\), i.e., 
%\[\lim_{n \to \infty} \frac{1}{n} \beta^{*} = -D(Q||R). \] 
Therefore, an agent wanting to increase the supervisor's false classification rate should minimize the KL divergence from $R$ to $Q$.\par

The goal of the agent is to design a deceptive policy $Q$ that minimizes the expected path cost $\mathbb{E}_Q\left[C_{0:T}(X_{0:T}, U_{0:T-1})\right]$ \eqref{eq:path cost} while limiting the KL divergence $D\left(Q\|R\right)$ \eqref{eq: KL cost}. Using \eqref{eq:path cost} and \eqref{eq: KL cost}, we propose the following KL control problem for the synthesis of optimal deceptive policies for the agent: 

\begin{prob}[Synthesis of optimal deceptive policy]\label{Prob:KL control}
\begin{align}\label{eq:prob_KL_dc1_fixed_end}
&\min_{\{Q_{U_t|X_t}\}_{t=0}^{T-1}} \mathbb{E}_Q \sum_{t=0}^{T-1}  \Big\{ C_t(X_t, U_t) \\
& +\lambda D(Q_{U_t|X_t}(\cdot | X_t) \|R_{U_t|X_t}(\cdot | X_t))\Big\} + \mathbb{E}_Q C_T(X_T)\nonumber
\end{align}
where $\lambda$ is a positive weighting factor that balances the trade-off between the KL divergence and the path cost.
\end{prob}
 We explain the above KL control problem via the following example:  
\begin{ex}
Consider a drone that is contracted by a supervisor to perform a surveillance task over an area. The supervisor prefers the drone to operate at high speeds (policy \(R\)) to improve the efficiency of the surveillance. The operator of the drone, the agent, on the other hand, prefers the drone to operate in a battery-saving, safe mode (policy \(Q\)) to improve the longevity of the drone. The agent does not want to get detected by the supervisor and fired. Hence, the goal of the agent is to operate in a way that would balance the energy consumption (\(\mathbb{E}_Q\left[ C_{0:T}(X_{0:T}, U_{0:T-1})\right]\)) and the deviations from the behavior desired by the supervisor (\(D(Q||R)\)).
\end{ex}

\section{Synthesis of Optimal Deceptive Policies}

In this section, we solve Problem \ref{Prob:KL control} using backward dynamic programming and propose a policy synthesis algorithm based on path integral control. 

\subsection{Backward Dynamic Programming}
Notice that the cost function of Problem \ref{Prob:KL control} possesses the time-additive Bellman structure and, therefore, can be solved by utilizing the principle of dynamic programming \cite{bertsekas2012dynamic}. Define for each $t\in\mathcal{T}$ and $x_t\in\mathcal{X}_t$, the value function:
\begin{align}\label{eq:value function}
&J_t(x_t):=\inf_{\{Q_{U_k|X_k}\}_{k=t}^{T-1}} \mathbb{E}_Q \sum_{k=t}^{T-1}  \Big\{ C_k(X_k, U_k)  \\
& \quad +\lambda D(Q_{U_k|X_k}(\cdot | X_k) \|R_{U_k|X_k}(\cdot | X_k))\Big\} + \mathbb{E}_Q C_T(X_T).\nonumber
\end{align}
Notice that in \eqref{eq:value function}, we used ``inf" instead of ``min" since we do not know if the infimum is attained. In the following theorem, we show that the infimum is indeed attained, and therefore, ``inf" can be replaced by ``min". 
\begin{theorem}\label{thrm:Bellman recursion}
The value function $J_t(x_t)$ satisfies the following backward Bellman recursion with the terminal condition $J_T(x_T)=C_T(x_T)$:
\begin{align}\label{eq:v_bellman} 
J_t(x_t)=& -\lambda \log\Bigg\{ \int_{\mathcal{U}_t}\exp\left(-\frac{C_t(x_t, u_t)}{\lambda}\right)\\ 
&\!\!\!\!\!\!\!\!\!\!\!\!\!\!\!\!\!\!\!\!\!\!\!\!\times\exp\left(-\frac{1}{\lambda}\int_{\mathcal{X}_{t+1}}\!\!\!\!\!\!J_{t+1}(x_{t+1})P(dx_{t+1}|x_t, u_t)\right) R(du_t|x_t)\Bigg\}\nonumber
\end{align}
and the minimizer of Problem \ref{Prob:KL control} is given by
\begin{equation}
\label{eq:p_star_dc}
Q_{U_t|X_t}^*(B_{U_t}|x_t)\!=\!\frac{\int_{B_{U_t}}\!\!\!\exp(-\rho_t(x_t, u_t)/\lambda)R(du_t|x_t)}{\int_{\mathcal{U}_t} \exp(-\rho_t(x_t, u_t)/\lambda)R(du_t|x_t)}
\end{equation}
where 
\begin{equation}\label{eq:rho_dc_def}
\!\!\!\rho_t(x_t, u_t)\!:=\! C_t(x_t, u_t)+\!\!\int_{\mathcal{X}_{t+1}} \!\!\!\!\!\!\!\!\! J_{t+1}(x_{t+1})P(dx_{t+1}|x_t, u_t)
\end{equation}
and $B_{U_t}$ is a Borel set belonging to the $\sigma-$algebra $\mathcal{B}(\mathcal{U}_t)$.
\end{theorem}

\begin{proof}
    See Appendix C.
\end{proof}
Theorem \ref{thrm:Bellman recursion} provides a recursive method to compute the value functions $J_t(x_t)$ and optimal control distributions $Q^*_{U_t|X_t}$ backward in time. As one can see, to perform the backward recursions \eqref{eq:v_bellman} and \eqref{eq:p_star_dc}, the function $J_t(x_t)$ must be evaluated everywhere in the continuous domain $\mathcal{X}_t$. Therefore, in practice, an exact implementation of backward dynamic programming is computationally costly (unless the problem has a special structure, for example, linear state dynamics and quadratic costs). The computational cost grows quickly with the dimension of the state space of the system, which is referred to as the \textit{curse of dimensionality}. In the next section (Section \ref{sec:PI_control}), we show that under the assumption of the deterministic state transition law, the backward Bellman recursions can be linearized. This allows us to design a simulator-driven algorithm to compute optimal deceptive actions.

\subsection{Simulator-Driven Control via Path Integral Approach}\label{sec:PI_control}
In this section, we focus on a special case in which the agent's dynamics are deterministic and propose a simulator-driven algorithm to compute the optimal deceptive actions via path integral control. 
\begin{assumption}\label{assump: deterministic law}
The state transition law is governed by a deterministic mapping $F_t:\mathcal{X}_t\times\mathcal{U}_t\rightarrow\mathcal{X}_{t+1}$ as
\begin{equation}\label{eq:deter state transition law}
    x_{t+1} = F_t(x_t, u_t);
\end{equation}
that is, $P(dx_{t+1}|x_t,u_t) = \delta_{F_t(x_t, u_t)}(dx_{t+1})$, where $\delta$ denotes the Dirac measure.
\end{assumption}

\begin{rem}
Note that under Assumption \ref{assump: deterministic law}, the agent can deviate from the reference policy \(R\) only if it is stochastic. Otherwise, under any deviations from the reference policy, with a positive probability, the supervisor will be sure that the agent did not follow the reference policy. Therefore, in what follows, we consider the reference policy to be stochastic. The stochasticity of the reference policy could be to account for the unmodeled elements of the dynamics, to provide robustness, or to encourage exploration. 
\end{rem}

\begin{rem}
Consider a special setting in which the state dynamics $F_t(x_t, u_t)$ is linear in $x_t$ and $u_t$, the reference policy distribution $R_{U_t|X_t}(\cdot | x_t)$ is Gaussian, and the cost functions $C_t(\cdot,\cdot)$ and $C_T(\cdot)$ are quadratic in $x_t$ and $u_t$. In such a setting, it can be shown that the optimal deceptive policy $Q^*_{U_t|X_t}(\cdot | x_t)$ is also Gaussian and can be analytically computed by backward Riccati recursions similar to the standard Linear-Quadratic-Regular (LQR) problems. In this work, we consider a broader setting with possibly non-Gaussian reference distribution, non-linear state dynamics, and non-quadratic cost functions. In this case, the optimal deceptive policy might not be efficiently computed by solving backward recursions.
\end{rem}

Now, we propose a path-integral-based solution approach for simulator-driven policy synthesis. Under assumption \ref{assump: deterministic law}, we can rewrite \eqref{eq:v_bellman} as
\begin{align}\label{eq:v_bellman_post}
J_t(x_t)=& -\lambda \log\Bigg\{ \int_{\mathcal{U}_t}\exp\left(-\frac{C_t(x_t, u_t)}{\lambda}\right)\\ 
&\!\!\!\!\!\!\!\!\!\!\!\!\!\!\!\!\!\!\!\!\!\!\!\!\times\exp\left(-\frac{J_{t+1}\left(F_t(x_t,u_t)\right)}{\lambda}\right)R(du_t|x_t)\Bigg\}.\nonumber
\end{align}
We introduce the exponentiated value function as $Z_t(x_t):=\exp\left(-\frac{1}{\lambda}J_t(x_t)\right)$. Using $Z_t(x_t)$, the Bellman recursion \eqref{eq:v_bellman_post} can be linearized, and we get the following linear relationship between $Z_t$ and $Z_{t+1}$:
\begin{align}
\! Z_t(x_t)\!=\!&\int_{\mathcal{U}_t}\!\!\!\exp\left(-\frac{C_t(x_t, u_t)}{\lambda}\right) Z_{t+1}\left(F_t(x_t, u_t)\right)R(du_t|x_t)\nonumber\\
=&\int_{\mathcal{U}_t}\int_{\mathcal{X}_{t+1}}\exp\left(-\frac{C_t(x_t, u_t)}{\lambda}\right) Z_{t+1}(x_{t+1})\label{eq:z_recursive}\\
&\hspace{15mm}\times P(dx_{t+1}|x_t, u_t)R(du_t|x_t).\nonumber
\end{align}
Note that in \eqref{eq:z_recursive}, $P(dx_{t+1}|x_t,u_t) = \delta_{F_t(x_t, u_t)}(dx_{t+1})$ by Assumption \ref{assump: deterministic law}. 
Equation \eqref{eq:z_recursive} is a linear backward recursion in \(Z_t\). The linear solvability of the KL control problem is well-known in the literature (e.g., \cite{todorov2007linearly}). We remark that linearizability critically relies on Assumption \ref{assump: deterministic law}\footnote{We remark that in prior works where the path integral method is used to solve stochastic control problems, a certain assumption (e.g. Eq. (9) in \cite{kappen2005path}) is made to reinterpret the original problem as a problem of designing the optimal randomized policy for a deterministic transition system.}.
%The original Bellman recursion \eqref{eq:v_bellman} cannot be linearized without Assumption \ref{assump: deterministic law}. 
Now, by recursive substitution, \eqref{eq:z_recursive} can also be written as
\begin{align*}
Z_t(x_t)=&
\int_{\mathcal{U}_t}\int_{\mathcal{X}_{t+1}}\cdots 
\int_{\mathcal{U}_{T-1}}\int_{\mathcal{X}_T}
\exp\left(-\frac{C_t(x_t, u_t)}{\lambda}\right)  \nonumber \\
&\!\!\!\!\!\!\!\!\!\!\!\!\times \cdots \times \exp\left(-\frac{C_T(x_T)}{\lambda}\right)
R(dx_{t+1:T}\times du_{t:T-1}|x_t). 
\end{align*}
Thus, by introducing the path cost function \[C_{t:T}(x_{t:T}, u_{t:T-1}):=\sum_{k=t}^{T-1}C_k(x_k, u_k)+C_T(x_T),\]
we obtain
\begin{equation}
  Z_t(x_t)=\mathbb{E}_R \exp \left(-\frac{1}{\lambda}C_{t:T}(X_{t:T}, U_{t:T-1}) \right)  
 \label{eq:z_pi}
\end{equation}
where the expectation 
$\mathbb{E}_R(\cdot)$ is with respect to the probability measure $R$ \eqref{eq:def_ref_traj_dist}. Equation \eqref{eq:z_pi} expresses the exponentiated value function $Z_t(x_t)$ as the expected path cost under the reference distribution. This suggests a path-integral-based approach to numerically compute $Z_t(x_t)$. Suppose we generate a collection of $N$ independent samples of system paths $\{x_{t:T}(i), u_{t:T-1}(i)\}_{i=1}^N$ starting from $x_t$ under the reference distribution $R$. Since the reference distribution is known, a collection of such sample paths can be easily generated using a Monte Carlo simulation. If $C_{t:T}(x_{t:T}(i), u_{t:T-1}(i))$ represents the path cost of the sample path $i$, then by the strong law of large numbers \cite{durrett2019probability} as $N\rightarrow \infty$, we get 
\begin{equation}
\label{eq:mc_dc_phi_convergence}
 \!\! \!\!\frac{1}{N}\sum_{i=1}^N \exp\!\left(-\frac{1}{\lambda}C_{t:T}(x_{t:T}(i), u_{t:T-1}(i))\right) \overset{a.s.}{\rightarrow}   Z_t(x_t).
\end{equation}\par
Similarly, a collection of sample paths $\{x_{t:T}(i), u_{t:T-1}(i)\}_{i=1}^N$ starting from $x_t$ under the reference distribution $R$ can be used to sample $u_t$ from the optimal distribution $Q_{U_t|X_t}^*(\cdot|x_t)$.
Notice that the optimal deceptive policy \eqref{eq:p_star_dc} can be expressed in terms of $\{Z_t\}_{t=0}^{T}$ as 
\begin{subequations}\label{eq:p_star_dc_linear_path_cost}
\begin{align}
 Q_{U_t|X_t}^*(B_{U_t}|x_t) =&
\frac{1}{Z_t(x_t)}\int_{B_{U_t}}\exp\left(-\frac{C_t(x_t, u_t)}{\lambda}\right)\nonumber\\ 
&\!\!\!\!\!\!\!\!\times\exp\left(-\frac{J_{t+1}\left(F_t(x_t,u_t)\right)}{\lambda}\right)R(du_t|x_t)\nonumber\\
&\!\!\!\!\!\!\!\!\!\!\!\!\!\!\!\!\!\!\!\!\!\!\!\!\!\!\!\!\!\!\!\!\!\!\!\!\!\!\!\!\!\!\!\!\!\!\!\!\!=\frac{1}{Z_t(x_t)}\int_{B_{U_t}}\exp\left(-\frac{C_t(x_t, u_t)}{\lambda}\right)\nonumber\\ 
&\!\!\!\!\!\!\!\!\!\!\!\!\!\!\!\!\times Z_{t+1}\left(F_t(x_t, u_t)\right)R(du_t|x_t)\label{eq:p_star_dc_linear_path_cost1}\\
&\!\!\!\!\!\!\!\!\!\!\!\!\!\!\!\!\!\!\!\!\!\!\!\!\!\!\!\!\!\!\!\!\!\!\!\!\!\!\!\!\!\!\!\!\!\!\!\!\!=\!\frac{1}{Z_t(x_t)}\int_{B_{U_t}}\int_{\mathcal{X}_{t+1}}\!\!\!\!\exp\left(-\frac{C_t(x_t, u_t)}{\lambda}\right)Z_{t+1}(x_{t+1})  \nonumber \\
&\times P(dx_{t+1}|x_t, u_t)R(du_t|x_t) \label{eq:p_star_dc_linear_path_cost2} \\
&\!\!\!\!\!\!\!\!\!\!\!\!\!\!\!\!\!\!\!\!\!\!\!\!\!\!\!\!\!\!\!\!\!\!\!\!\!\!\!\!\!\!\!\!\!\!\!\!\!=\!\frac{1}{Z_t(x_t)}\!\!\int_{\!\{\mathcal{X}_{t+1:T},\,\mathcal{U}_{t:T-1}|u_t\in B_{U_t}\}} \!\!\!\!\!\!\! \exp\!\left(\!-\frac{C_{t:T}(x_{t:T}, u_{t:T-1})}{\lambda}\!\right) \nonumber \\
&\hspace{3ex}\times R(dx_{t+1:T}\times du_{t:T-1}|x_t).
\label{eq:p_star_dc_linear_path_cost3}
\end{align}
\end{subequations}
The step \eqref{eq:p_star_dc_linear_path_cost1} follows from the definition of $Z_t$. In \eqref{eq:p_star_dc_linear_path_cost2}, we used our assumption $P(dx_{t+1}|x_t,u_t) = \delta_{F_t(x_t, u_t)}(dx_{t+1})$. Finally, \eqref{eq:p_star_dc_linear_path_cost3} is obtained by the recursive substitution of \eqref{eq:p_star_dc_linear_path_cost2}, and $\{\mathcal{X}_{t+1:T},\,\mathcal{U}_{t:T-1}|u_t\in B_{U_t}\}$ represents a collection paths such that $u_t\in B_{U_t}$. \par
% \mustafa{Maybe a connecting sentence here. Like "We use the above representation of \(Q^{*}\) to sample an action \(u_{t}\). Note that in the above integral control actions leading to paths with high costs have exponentially less mass/weight and the probabilities of control actions under \(Q^{*}\) are proportional to probabilities under the reference policy \(R\). In the next part, we construct an algorithm that combines these two features."}
We use the above representation of \(Q^{*}_{U_t|X_t}\) to sample an action \(u_{t}\) from it. Let $r_t(i)$ be the exponentiated path cost of the sample path $i$:
\begin{equation}
\label{eq:dc_path_reward}
r_t(i):= \exp\left(-\frac{1}{\lambda}C_{t:T}(x_{t:T}(i), u_{t:T-1}(i))\right)
\end{equation}
and $r_t:=\sum_{i=1}^N r_t(i)$.
For each $t\in\mathcal{T}$, we introduce a piecewise constant, monotonically non-decreasing function $F_t: [0,N]\rightarrow [0,r_t]$ by
\[
F_t(x)=\sum_{i=1}^{\lfloor x \rfloor} r_t(i).
\]
where $\lfloor x \rfloor$ denotes $\mathrm{floor}(x)$, i.e., the greatest integer less than or equal to $x$. The function $F_t(x)$ is represented in Figure \ref{fig:f_func}.
% \begin{figure}[t]
% \begin{tikzpicture}
%  \node{\includegraphics[width=\columnwidth]{f_func}};
%   \node[text width=3cm] at (-2.35,-1.87) 
%     {\tiny $t$};
%       \node[text width=3cm] at (-2.35,-1.32) 
%     {\tiny $t$};
%       \node[text width=3cm] at (-2.30,1.72) 
%     {\tiny $t$};
%       \node[text width=3cm] at (-1.3,2.23) 
%     {\tiny $t$};
% \end{tikzpicture}
%     \caption{Function $F_{t}(x)$.}
%     \label{fig:f_func}
% \end{figure}

\begin{figure}[t]
{\includegraphics[width=\columnwidth]{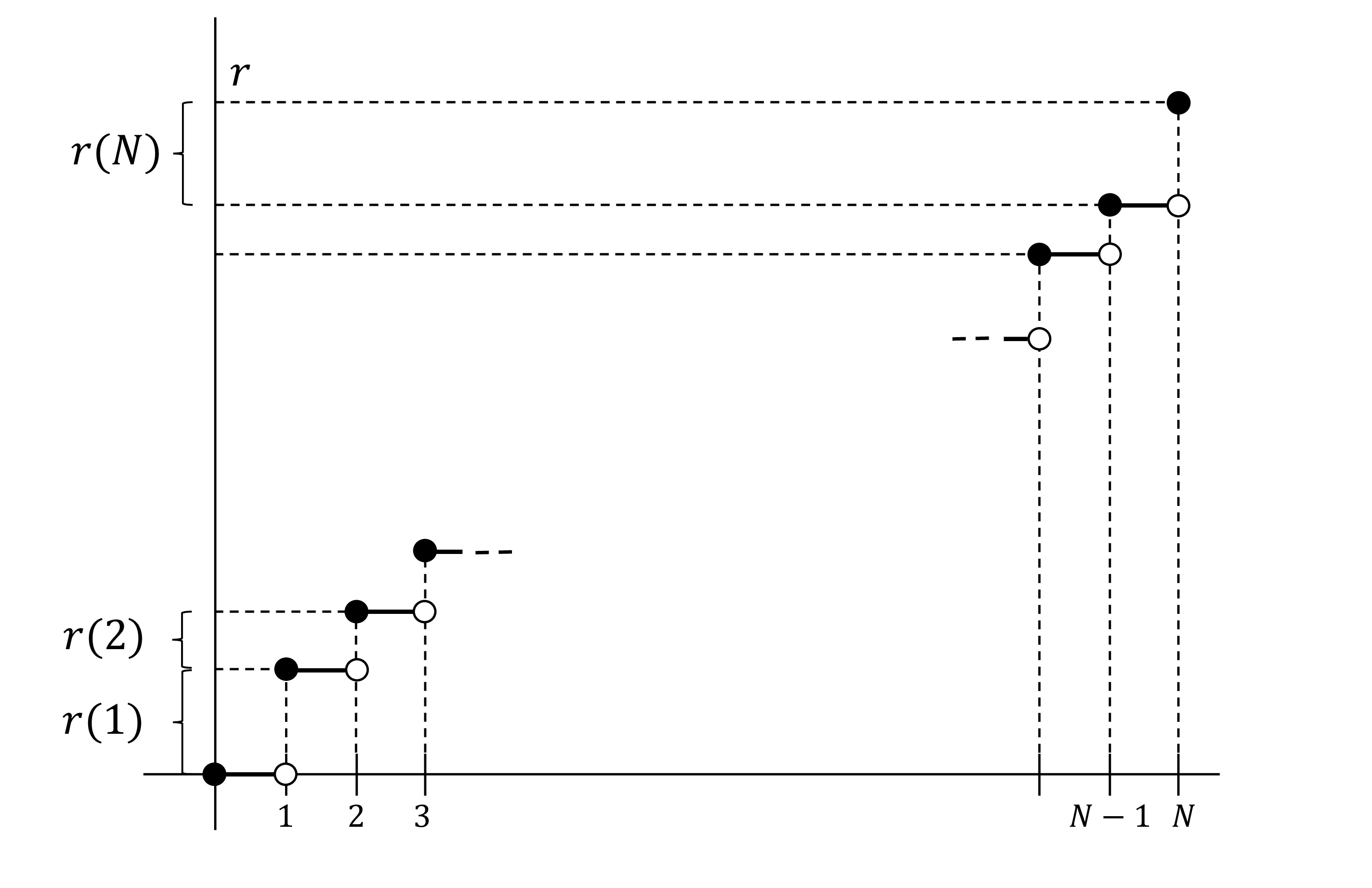}};
    \caption{Function $F_{t}(x)$.}
    \label{fig:f_func}
\end{figure}

\begin{algorithm}[t]\label{Algo: Q*}
\caption{Sampling $u_t$ approximately from $Q_{U_t|X_t}^*(\cdot|x_t)$ by Monte Carlo simulations}
\KwData{Initial state $x_0$}
\For {$t\in\mathcal{T}$}{
Sample $N$ paths $\{x_{t:T}(i), u_{t:T-1}(i)\}_{i=1}^N$ starting from $x_t$ under the reference distribution $R$. \label{alg:randomsamplepath}
\\
For each sample path $i$, compute the exponentiated path cost $r_t(i)$ by \eqref{eq:dc_path_reward}.\label{alg:reward of a sample path}\\
Compute $r_t:=\sum_{i=1}^N r_t(i)$.\label{alg:total reward}\\
Generate a random variable $d$ according to $d \sim \mathrm{unif}[0,r_t]$. \label{alg:randomoutputpath}
\\
Select a sample ID by $j_t\leftarrow F_t^{-1}(d)$.  \label{alg:sample ID}
\\ 
Select a control input as $u_t\leftarrow u_t(j_t)$.
}
\end{algorithm}

Notice that the inverse $F_t^{-1}$ of $F_t$ defines a mapping $F_t^{-1}: [0,r_t]\rightarrow \{1, 2, ... , N\}$.
To generate a sample $u_t$ approximately from the optimal distribution $Q_{U_t|X_t}^*$, we propose Algorithm \ref{Algo: Q*}.
We first, generate a random variable $d$ according to $d \sim \text{unif}[0,r_t]$. Then, we select a sample ID by $j_t\leftarrow F_t^{-1}(d)$. Finally, the control input adopted in the $j_t$-th sample path at time step $t$ is selected as $u_t$, i.e., $u_t\leftarrow u_t(j_t)$. Theorem \ref{theorem:lln} proves that as the number of Monte Carlo samples tends to infinity, Algorithm 1 samples $u_t$ from the optimal distribution $Q_{U_t|X_t}^*(\cdot|x_t)$.

\begin{theorem}
\label{theorem:lln}
Let $B_{U_t}\in\mathcal{B}(\mathcal{U}_t)$ be a Borel set. Suppose for a given collection of sample paths $\{x_{t:T}(i), u_{t:T-1}(i)\}_{i=1}^N$, $u_t$ is computed by Algorithm 1 and the probability of $u_t\in B_{U_t}$ is denoted by $\Pr\{u_t \in B_{U_t} |\{x_{t:T}(i),  u_{t:T-1}(i)\}_{i=1}^N \}$. Then, as $N \rightarrow \infty$
\[
\Pr\{u_t \in B_{U_t} |\{x_{t:T}(i),  u_{t:T-1}(i)\}_{i=1}^N \} \overset{a.s.}{\rightarrow} Q^*_{U_t|X_t}(B_{U_t}|x_t). 
\]
\end{theorem}
\begin{proof}
    See Appendix D 
\end{proof}
\vspace{1mm}

We showed that under Assumption \ref{assump: deterministic law}, the optimal deceptive policies can be synthesized using path integral control. Algorithm \ref{Algo: Q*} allows the deceptive agent to numerically compute optimal actions via Monte Carlo simulations without explicitly synthesizing the policy. Since Monte Carlo simulations can be efficiently parallelized, the agent can generate the optimal control actions online.

\section{Numerical Example}
\begin{figure*}
     \centering
       \!\!\!\!\!\!\!\!\!\!\!\!\!\!\!\begin{tabular}{c c}
       \includegraphics[scale=0.75]{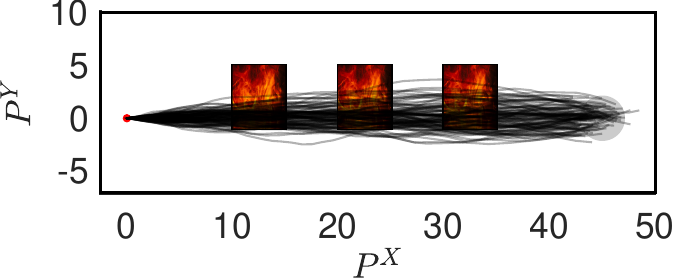} &\includegraphics[scale=0.75]{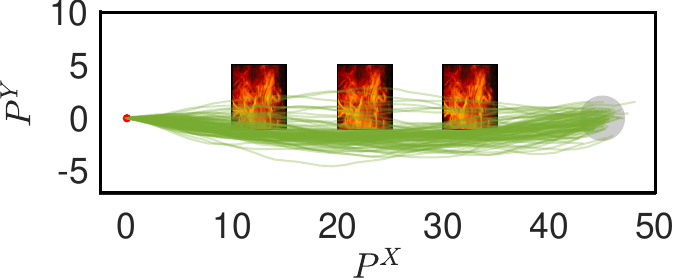} \\
       (a) Paths under $R$, $\Pr^{\mathrm{safe}} = 0.04$  & (b) Paths under $\widehat{Q}^*$ with $\lambda=3$, $\Pr^{\mathrm{safe}} = 0.48$ \\
       \vspace{1mm}\\
      \includegraphics[scale=0.75]{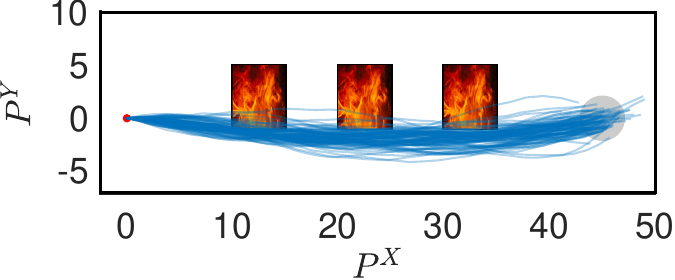} &\includegraphics[scale=0.75]{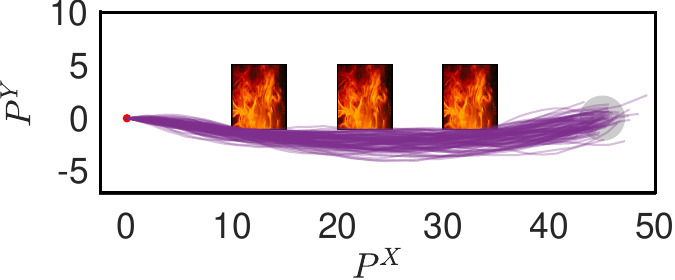} \\
       (c) Paths under $\widehat{Q}^*$ with $\lambda=2$, $\Pr^{\mathrm{safe}} = 0.62$  & (d) Paths under $\widehat{Q}^*$ with $\lambda=0.5$, $\Pr^{\mathrm{safe}} = 0.94$ \\
       \end{tabular}
         \caption{A unicycle navigation problem. The start position is shown by a red dot, and the goal region by a disk colored in gray. $100$ sample paths generated under the reference policy $R$ and the agent's policy $\widehat{Q}^*$ with three values of $\lambda$ are shown. The probability of safe paths $\Pr^{\mathrm{safe}}$ are noted below each case.} 
         \label{Fig. simulation trajs}
 \end{figure*}
In this section, we validate the path-integral-based algorithm proposed to generate optimal deceptive control actions. The problem is illustrated in Figure \ref{Fig. simulation trajs}. A supervisor wants an agent to start from the origin and reach a disk of radius $G^R$ centered at $\begin{bmatrix} G^X & G^Y\end{bmatrix}^\top$ (shown in gray color) as fast as possible. The supervisor also expects the agent to inspect the region on the way. To encourage exploration and to provide robustness against unmodeled dynamics, the supervisor designs a randomized reference policy. The agent, on the other hand, wishes to avoid the regions on the way that are covered under fire, as shown in Figure \ref{Fig. simulation trajs}. Let these regions be represented collectively by $\mathcal{X}^{\mathrm{fire}}$. Suppose the agent's dynamics are modeled by a unicycle model as:
\begin{align*}
    P^X_{t+1} &= P^X_{t+1} + S_t\cos\Theta_t h\\
    P^Y_{t+1} &= P^Y_{t+1} + S_t\sin\Theta_t h\\
    S_{t+1} & = S_{t} + A_t h\\
    \Theta_{t+1} & = \Theta_{t} + \Omega_t h
\end{align*}
where $(P^X_t, P^Y_t)$, $S_t$, and $\Theta_t$ denote the $x-y$ position, speed, and the heading angle of the agent at time step $t$, respectively. The control input $U_t\coloneqq\begin{bmatrix}A_t & \Omega_t\end{bmatrix}^T$ consists of acceleration $A_t$ and angular speed $\Omega_t$. $h$ is the time discretization parameter used for discretizing the continuous-time unicycle model. For this simulation study, we set $h=1$. Note that the agent's dynamics is deterministic as per Assumption \ref{assump: deterministic law}; however, the control input $U_t$ can be stochastic. Suppose the supervisor designs the reference policy $R$ as a Gaussian probability density with mean $\overline{u}_t$ and covariance $\Sigma_t$:
\begin{equation*}
  R_{U_t|X_t}(\cdot|x_t) = \frac{\exp\left[-\frac{1}{2}(u_t-\overline{u}_t)^\top\Sigma_t^{-1}(u_t-\overline{u}_t)\right]}{\sqrt{(2\pi)^2|\Sigma_t|}}.
\end{equation*}
The mean $\overline{u}_t\coloneqq\begin{bmatrix}\overline{a}_t & \overline{\omega}_t \end{bmatrix}^\top$ is designed using a proportional controller as 
\begin{equation*}
    \overline{A}_t = -k_{A} (S_t - S_t^{\mathrm{desired}}),\quad
    \overline{\Omega}_t = -k_{\Omega} (\Theta_t - \Theta_t^{\mathrm{desired}}) 
\end{equation*}
where $k_{A}$ and $k_{\Omega}$ are proportional gains and $S_t^{\mathrm{desired}}$, $\Theta_t^{\mathrm{desired}}$ are computed as 
\begin{equation*}
  S_t^{\mathrm{desired}} \!=\! \frac{\norm{\begin{bmatrix}G^X\\G^Y\end{bmatrix}\!  -\! \begin{bmatrix}P^X_t\\P^Y_t\end{bmatrix}}}{T-t} ,\;
\Theta_t^{\mathrm{desired}}\!  =\! \tan^{\!-\!1}\!\!\left(\!\frac{G^Y-P^Y_t}{G^X-P^X_t}\!\right).  
\end{equation*}
As mentioned before, the agent wishes to avoid the region $\mathcal{X}^{\mathrm{fire}}$. Suppose the cost function $C_{0:T}$ is designed as
\begin{equation*}
   C_{0:T}(X_{0:T}, U_{0:T-1}) = \sum_{t=0}^{T}\mathds{1}_{[P^X_t\;P^Y_t]^{^\top}\in\mathcal{X}^{\mathrm{fire}}}
\end{equation*}
where $\mathds{1}_{[P^X_t\;P^Y_t]^{^\top}\in\mathcal{X}^{\mathrm{fire}}}$ represents an indicator function that returns $1$ when the agent is inside the region $\mathcal{X}^{\mathrm{fire}}$ and $0$ otherwise.
For this simulation, we set
\begin{equation*}
   \begin{bmatrix}G^X\\
    G^Y\end{bmatrix}\!\! =\!\! \begin{bmatrix}45\\
    0\end{bmatrix}\!\!,\, \Sigma_t \!\!=\!\! \begin{bmatrix}0.5 & \!\!\!\!0\\
    0 & \!\!\!\!0.5\end{bmatrix}\!\!,\, k_A\! =\! 0.1, \, k_\Omega\! =\! 0.2, \, T\! =\! 50.
\end{equation*}
The agent chooses its action at each time step using Algorithm \ref{Algo: Q*} where the number of samples is $N=10^5$. \par
 
Suppose $\widehat{Q}^*$ denotes the deceptive agent's distribution generated by the sampling-based Algorithm 1. Figure \ref{Fig. simulation trajs} shows $100$ paths under the reference distribution R (Figure \ref{Fig. simulation trajs}(a)) and the agent's distribution $\widehat{Q}^*$ for three values of $\lambda$ (Figure \ref{Fig. simulation trajs}(b) - \ref{Fig. simulation trajs}(d)). A lower value of $\lambda$ implies that the agent cares less about its deviation from the reference policy and more about avoiding the region $\mathcal{X}^{\mathrm{fire}}$. A higher value of $\lambda$ implies the opposite. We also report $\Pr^{\mathrm{safe}}$, the percentage of paths that avoid $\mathcal{X}^{\mathrm{fire}}$. Under the reference distribution $R$, only $4\%$ of the paths are safe. On the other hand, more paths are safe under the agent's distribution $\widehat{Q}^*$, and as the value of $\lambda$ reduces, $\Pr^{\mathrm{safe}}$ increases.

 \par
 Figure \ref{Fig. LLR} shows the expected log-likelihood ratio (with one standard deviation) with respect to time $t$ for three values of $\lambda$. The expected LLR is computed as follows. Algorithm \ref{Algo: Q*} selects a control input $u_k\leftarrow u_k(j_k)$ at time step $k$, where $j_k$ is a sample ID obtained from step \ref{alg:randomoutputpath}. From the construction of the algorithm, at each time step $k$, the probability of choosing the control input $u_k\leftarrow u_k(j_k)$ under the agent's distribution $\widehat{Q}^*$ is $r_k(j_k)/r_k$, where $r_k(j_k)$ and $r_k$ are computed by steps \ref{alg:reward of a sample path} and \ref{alg:total reward} of Algorithm \ref{Algo: Q*}. Whereas the probability of choosing the control input $u_k\leftarrow u_k(j_k)$ under the reference distribution is $1/N$. Therefore, using \eqref{eq:llr and KL}, the expected LLR upto time $t\in\mathcal{T}$ can be approximately computed as
\begin{align*}
     &\mathbb{E}_{Q^*}\left[\log \frac{dQ^*_{X_{0:t}\times U_{0:t-1}}}{dR_{X_{0:t}\times U_{0:t-1}}}(x_{0:t}, u_{0:t-1})\right]\nonumber\\
     =& \mathbb{E}_{Q^*}\!\!\! \left[\sum_{k=0}^{t-1}\log\frac{dQ^*_{U_k|X_k}}{dR_{U_k|X_k}}\left(x_k, u_k\right)\right]\!\!\approx\!\frac{1}{N_{\widehat{Q}^*}}\!\!\sum_{i=1}^{N_{\widehat{Q}^*}}\sum_{k=0}^{t-1} \frac{r_k(j_k)/r_k}{1/N}.
\end{align*}
where $N_{\widehat{Q}^*}$ is the number of paths generated by repeatedly running Algorithm \ref{Algo: Q*}. Note that since we assume the system dynamics to be deterministic (Assumption \ref{assump: deterministic law}), once the control input $u_k$ is chosen at time step $k$, the state $x_{k+1}$ is uniquely determined. Therefore, while computing the expected LLR, we only need to consider the probabilities of choosing the control input $u_k$ under policies $\widehat{Q}^*$ and $R$. Figure \ref{Fig. LLR} shows that for a lower value of $\lambda$, the expected LLR is higher, i.e., more deviation of $\widehat{Q}^*$ from $R$.

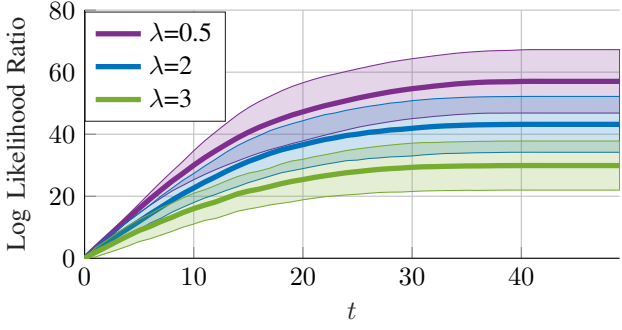
\begin{figure}
     \centering
      \input{log_like}
      %\!\!\!\!\!\!\!\!\!\!\!\begin{tabular}{c}
      %\includegraphics[scale=0.5]{log_like.pdf} 
      %\end{tabular}
         \caption{Expected LLR (with one standard deviation) with respect to time $t$ for three values of $\lambda$.} 
         \label{Fig. LLR}
 \end{figure}

\section{Conclusion}
We presented a deception problem under supervisory control for continuous-state discrete-time stochastic systems. Using motivations from hypothesis testing theory, we formalized the synthesis of an optimal deceptive policy as a KL control problem and solved it using backward dynamic programming. Since the dynamic programming approach suffers from the curse of dimensionality, we proposed a simulator-driven algorithm to compute optimal deceptive actions via path
integral control. The proposed approach allows the agent to numerically compute deceptive actions online via Monte Carlo sampling of system paths. We validated the proposed approach via a numerical example with a nonlinear system.
%We also showed that the proposed method has a polynomial sample complexity rate for a special class of binary cost functions.\par

For future work, we plan to study the deception problem for continuous-time stochastic systems. We also plan to conduct the sample complexity analysis of the path integral approach to solve KL control problems.


%%%%%%%%%%%%%%%%%%%%%%%%%%%%%%%%%%%%%%%%%%%%%%%%%%%%%%%%%%%%%%%%%%%%%%%%%%%%%%%%

%%%%%%%%%%%%%%%%%%%%%%%%%%%%%%%%%%%%%%%%%%%%%%%%%%%%%%%%%%%%%%%%%%%%%%%%%%%%%%%%

\appendix
%\section*{APPENDIX}
\subsection{Proof of Equation \eqref{eq: KL cost}}
\vspace{-5mm}
\begin{subequations}\label{eq:RN as product} \allowdisplaybreaks
\begin{align}
   &\int_{B} Q_{X_{0:T}\times U_{0:T-1}}(dx_{0:T}\!\times\! du_{0:T-1})\nonumber\\
   =&\int_{B}\prod_{t=0}^{T-1}\!\! P(dx_{t+1}|x_t, u_t)Q(du_t|x_t)\label{eq:RN as product1}\\
   =&\int_{B}\!\!\left(\prod_{t=0}^{T-1}\frac{dQ_{U_t|X_t}}{dR_{U_t|X_t}}\!\left(x_t,u_t\right)\!\!\right)\!\!\prod_{t=0}^{T-1}\!\! P(dx_{t+1}|x_t, u_t)R(du_t|x_t)\label{eq:RN as product2}\\
   =&\!\!\!\int_{B}\!\!\left(\prod_{t=0}^{T-1}\frac{dQ_{U_t|X_t}}{dR_{U_t|X_t}}\!\left(x_t,u_t\right)\!\!\right)\!\!R_{X_{0:T}\times U_{0:T-1}}\!(dx_{0:T}\!\!\times\! du_{0:T-1})\label{eq:RN as product3}
\end{align}
\end{subequations}
where, $B$ is a Borel set belonging to the $\sigma-$algebra $\mathcal{B}\left(\mathcal{X}_{0:T}\times\mathcal{U}_{0:T-1}\right)$. The first equality \eqref{eq:RN as product1} follows from the definition \eqref{eq:def_dc_traj_dist}, the second equality \eqref{eq:RN as product2} by the definition of the Radon-Nikodym derivative \cite{durrett2019probability} and the last one \eqref{eq:RN as product3} from the definition \eqref{eq:def_ref_traj_dist}. Using \eqref{eq:RN as product}, we can write the Radon-Nikodym derivative $\frac{dQ_{X_{0:T}\times U_{0:T-1}}}{dR_{X_{0:T}\times U_{0:T-1}}}$ as follows:
\begin{equation}\label{eq: stagewise KL}
     \!\!\!\frac{dQ_{X_{0:T}\times U_{0:T-1}}}{dR_{X_{0:T}\times U_{0:T-1}}}\left(x_{0:T}, u_{0:T-1}\right) \!=\!\!\! \prod_{t=0}^{T-1} \frac{dQ_{U_t|X_t}}{dR_{U_t|X_t}}\left(x_t,u_t\right).
\end{equation}
Using \eqref{eq: stagewise KL}, we get the following:

\begin{align}
 D(Q\|P)=&\;\mathbb{E}_Q\left[\log \frac{dQ_{X_{0:T}\times U_{0:T-1}}}{dR_{X_{0:T}\times U_{0:T-1}}}(x_{0:T}, u_{0:T-1})\right]\nonumber\\
     =&\;\mathbb{E}_Q \left[\log \prod_{t=0}^{T-1} \frac{dQ_{U_t|X_t}}{dR_{U_t|X_t}}\left(x_t, u_t\right)\right]\nonumber\\
     =&\;\mathbb{E}_Q\left[\sum_{t=0}^{T-1}\log\frac{dQ_{U_t|X_t}}{dR_{U_t|X_t}}\left(x_t, u_t\right)\right]\label{eq:llr and KL}\\
     =&\;\mathbb{E}_Q\left[\sum_{t=0}^{T-1}D(Q_{U_t|X_t}(\cdot | X_t) \|R_{U_t|X_t}(\cdot | X_t))\right].\nonumber
\end{align}

\subsection{Legendre Duality}
\label{appendix:a}

Let $P$ and $Q$ be probability distributions on $(\mathcal{X}, \mathcal{B}(\mathcal{X}))$, and $C:\mathcal{X}\rightarrow \mathbb{R}$ a given cost function. Define the internal energy $U(P,C)$, free energy $F(R,C)$ and relative entropy (KL divergence) $D(P\|R)$ as:
\begin{align*}
&U(P,C):=\int_\mathcal{X} C(x)P(dx)  \\
&F(R,C):= -\lambda \log \int_\mathcal{X} \exp\left(-\frac{C(x)}{\lambda}\right)R(dx) \\
&D(P\|R):= \int_\mathcal{X} \log\frac{dP}{dR}(x)P(dx).
\end{align*}
Then the following duality relationship holds:
\begin{align*}
F(R,C)&=\inf_P \{ U(P,C)+\lambda D(P\|R) \} \\
    -\lambda D(P\|R)&=\inf_C \{ U(P,C)-F(R,C) \}.
\end{align*}
Also, the optimal probability distribution $P^*$ is given by
\begin{equation*}
    P^*(B) = \frac{\int_{B}\exp(-C(x)/\lambda)R(dx)}{\int_{\mathcal{X}} \exp(-C(x)/\lambda)R(dx)}, \quad \forall B \in \mathcal{B}(\mathcal{X}).
\end{equation*}
See \cite{boue1998variational,theodorou2012relative} for further discussions.

\subsection{Proof of Theorem 1}
\label{appendix:c}
By Bellman's optimality principle, the value function satisfies the following recursive relationship:
\begin{equation}\label{eq:kl_dc_bellman1}
J_t(x_t)\!=\!\!\!\!\!\inf_{Q_{U_t|X_t}}\! \int_{\mathcal{U}_t}\!\!\! \left\{\!\rho_t(x_t, u_t)\! +\!\lambda \log \frac{dQ}{dR}(u_t|x_t)\!\right\}\!Q(du_t|x_t) 
\end{equation}
where $\rho_t(x_t, u_t)$ is defined by \eqref{eq:rho_dc_def}. Invoking the Legendre duality between the KL divergence and free energy (see Appendix B), it can be shown that there exists a minimizer $Q^*_{U_t|X_t}$ of the right-hand side of \eqref{eq:kl_dc_bellman1}, which can be written as 
\begin{equation}\label{eq:Q_bellman}
Q_{U_t|X_t}^*(B_{U_t}|x_t)\!=\!\frac{\int_{B_{U_t}}\!\!\!\exp(-\rho_t(x_t, u_t)/\lambda)R(du_t|x_t)}{\int_{\mathcal{U}_t} \exp(-\rho_t(x_t, u_t)/\lambda)R(du_t|x_t)}
\end{equation}
where $B_{U_t}$ is a Borel set belonging to the $\sigma-$algebra $\mathcal{B}(\mathcal{U}_t)$. Using \eqref{eq:Q_bellman}, the value of \eqref{eq:kl_dc_bellman1} can be computed as
\begin{equation}
\label{eq:v_rho_bellman}
\!\!J_t(x_t)\!=\!-\lambda \log \left\{ \int_{\mathcal{U}_t}\!\! \exp\left(-\frac{\rho_t(x_t, u_t)}{\lambda}\right)\!R(du_t|x_t) \right\}.
\end{equation}
Substituting \eqref{eq:rho_dc_def} into \eqref{eq:v_rho_bellman}, we obtain the recursive expression \eqref{eq:v_bellman}.
\subsection{Proof of Theorem 2}
\label{appendix:d}
Let $\mathcal{I}_{B_{U_t}}$ be the set of indices of sample paths for which an action in $B_{U_t}$ is taken at time step $t$, i.e., $\mathcal{I}_{B_{U_t}}=\{i\in \{1, 2, \ldots , N\} | u_t(i)\in B_{U_t}\}$. Define the sum of the exponentiated path costs of the paths in $\mathcal{I}_{B_{U_t}}$ as
$r_{B_{U_t}}=\sum_{i\in\mathcal{I}_{B_{U_t}}} r_t(i).$
By construction of Algorithm \ref{Algo: Q*},
\begin{equation}\label{eq:rB/r}
    \Pr\{u_t \in B_{U_t} |\{x_{t:T}(i),  u_{t:T-1}(i)\}_{i=1}^N \}=\frac{r_{B_{U_t}}}{r_t}.
\end{equation}
Now, from \eqref{eq:mc_dc_phi_convergence}, as $N\rightarrow \infty$, we get
\[
\frac{r_t}{N}=\frac{1}{N}\sum_{i=1}^N \exp\left(-\frac{C_{t:T}(x_{t:T}(i), u_{t:T-1}(i))}{\lambda}\right)\overset{a.s.}{\rightarrow} Z_t(x_t). 
\]
Similarly, as $N\rightarrow \infty$,
\begin{align*}
&\frac{r_{B_{U_t}}}{N}=\frac{1}{N}\sum_{i\in\mathcal{I}_{B_{U_t}}} \exp\left(-\frac{C_{t:T}(x_{t:T}(i), u_{t:T-1}(i))}{\lambda}\right) \\
&\overset{a.s.}{\rightarrow} \int_{\{\mathcal{X}_{t+1:T},\; \mathcal{U}_{t:T-1}|u_t\in B_{U_t}\}} \exp\left(-\frac{C_{t:T}(x_{t:T},  u_{t:T-1})}{\lambda}\right)  \\
&\hspace{27ex}\times R(dx_{t+1:T},du_{t:T-1}|x_t)
\end{align*}
Thus, from \eqref{eq:p_star_dc_linear_path_cost3} and \eqref{eq:rB/r}
\begin{align*}
&\Pr\{u_{t}\in B_{U_t} |\{x_{t:T}(i),  u_{t:T-1}(i)\}_{i=1}^N\}\\
&\overset{a.s.}{\rightarrow}\!\! \frac{1}{Z_t(x_t)}\!\!\int_{\{\mathcal{X}_{t+1:T},\; \mathcal{U}_{t:T-1}|u_t\in B_{U_t}\}} \!\!\!\!\!\!\!\!\!\!\! \exp\left(-\frac{C_{t:T}(x_{t:T},  u_{t:T-1})}{\lambda}\right)  \\
&\hspace{28ex}\times R(dx_{t+1:T}, du_{t:T-1}|x_t)\\
&=Q_{U_t|X_t}^*(B_{U_t}|x_t).
\end{align*} 

%%%%%%%%%%%%%%%%%%%%%%%%%%%%%%%%%%%%%%%%%%%%%%%%%%%%%%%%%%%%%%%%%%%%%%%%%%%%%%%%\

\bibliographystyle{IEEEtran}
\bibliography{bibliography}

\end{document}

%% file: log_like.tex
\begin{tikzpicture}

\begin{axis}[%
width=2.8in,
height=1.3in,
at={(0in,0in)},
scale only axis,
xmin=0,
xmax=49,
xlabel style={font=\color{white!15!black}},
xlabel={$t$},
ymin=0,
ymax=80,
ylabel style={font=\color{white!15!black}},
ylabel={Log Likelihood Ratio},
axis background/.style={fill=white},
axis x line*=bottom,
axis y line*=left,
xmajorgrids,
ymajorgrids,
legend style={at={(0,0.55)}, anchor=south west, legend cell align=left, align=left, draw=white!15!black}
]

\addplot[area legend, draw=mypurple, fill=mypurple, fill opacity=0.2, forget plot]
table[row sep=crcr] {%
x	y\\
0   0\\
1	3.5523113582886\\
2	7.20496656647935\\
3	10.576139716813\\
4	14.078202679469\\
5	17.5845574074926\\
6	21.1540299643263\\
7	24.6275939214113\\
8	27.9759698955718\\
9	31.2314439614172\\
10	34.6790230395343\\
11	37.8732031172596\\
12	40.7277604950645\\
13	43.5156816762281\\
14	46.2314587265719\\
15	48.5445775952152\\
16	50.5518842586048\\
17	52.230399733868\\
18	53.9192396382854\\
19	55.377708276403\\
20	56.5820476218107\\
21	57.6891517464069\\
22	58.6699540934099\\
23	59.4310609309446\\
24	60.2608522721388\\
25	61.0454106975198\\
26	61.8051765409449\\
27	62.547683061639\\
28	63.2546173863369\\
29	63.7959277784103\\
30	64.3227853817026\\
31	64.7796699928003\\
32	65.2539091454383\\
33	65.6790346633522\\
34	66.0930666011181\\
35	66.3984734411413\\
36	66.6429747222277\\
37	66.8482309623852\\
38	66.9507296944311\\
39	67.0141720407168\\
40	67.1926642307962\\
41	67.269319004698\\
42	67.2504170922385\\
43	67.253134267294\\
44	67.253134267294\\
45	67.253134267294\\
46	67.253134267294\\
47	67.253134267294\\
48	67.253134267294\\
49	67.253134267294\\
49	46.808764862585\\
48	46.808764862585\\
47	46.808764862585\\
46	46.808764862585\\
45	46.808764862585\\
44	46.808764862585\\
43	46.808764862585\\
42	46.807353864224\\
41	46.8157682758702\\
40	46.8192830974176\\
39	46.8027805402073\\
38	46.7866490193912\\
37	46.784399554916\\
36	46.6908552038373\\
35	46.4966377215893\\
34	46.3056958135099\\
33	46.0036919080726\\
32	45.7242018791232\\
31	45.3402475008855\\
30	45.021197733381\\
29	44.5696221245557\\
28	44.0406360095486\\
27	43.4604226686005\\
26	42.8775133011866\\
25	42.1949060301525\\
24	41.3639764156124\\
23	40.5471359062307\\
22	39.6789213690612\\
21	38.7740358425081\\
20	37.8476024186979\\
19	36.9656979340807\\
18	35.969610368931\\
17	35.0473988661163\\
16	33.9377739289592\\
15	32.6973018048621\\
14	31.3585661900023\\
13	30.0116558033631\\
12	28.5719782109428\\
11	27.0079519215159\\
10	25.3291833683109\\
9	23.5049140307571\\
8	21.674751623475\\
7	19.4069333634583\\
6	17.0918754207325\\
5	14.5506767292436\\
4	11.9822749725825\\
3	9.06336326434564\\
2	5.96926520751336\\
1	3.14825720816768\\
0   0\\
}--cycle;
\addplot [color=mypurple, line width=2.0pt]
  table[row sep=crcr]{%
  0   0\\
1	3.35028428322814\\
2	6.58711588699636\\
3	9.81975149057933\\
4	13.0302388260258\\
5	16.0676170683681\\
6	19.1229526925294\\
7	22.0172636424348\\
8	24.8253607595234\\
9	27.3681789960872\\
10	30.0041032039226\\
11	32.4405775193877\\
12	34.6498693530036\\
13	36.7636687397956\\
14	38.7950124582871\\
15	40.6209397000387\\
16	42.244829093782\\
17	43.6388992999922\\
18	44.9444250036082\\
19	46.1717031052418\\
20	47.2148250202543\\
21	48.2315937944575\\
22	49.1744377312355\\
23	49.9890984185876\\
24	50.8124143438756\\
25	51.6201583638362\\
26	52.3413449210657\\
27	53.0040528651198\\
28	53.6476266979427\\
29	54.182774951483\\
30	54.6719915575418\\
31	55.0599587468429\\
32	55.4890555122807\\
33	55.8413632857124\\
34	56.199381207314\\
35	56.4475555813653\\
36	56.6669149630325\\
37	56.8163152586506\\
38	56.8686893569111\\
39	56.9084762904621\\
40	57.0059736641069\\
41	57.0425436402841\\
42	57.0288854782312\\
43	57.0309495649395\\
44	57.0309495649395\\
45	57.0309495649395\\
46	57.0309495649395\\
47	57.0309495649395\\
48	57.0309495649395\\
49	57.0309495649395\\
};
\addlegendentry{$\lambda\text{=0.5}$}

\addplot[area legend, draw=myblue, fill=myblue, fill opacity=0.2, forget plot]
table[row sep=crcr] {%
x	y\\
0   0\\
1	3.68974822977369\\
2	6.7972644520095\\
3	9.38543211259055\\
4	11.9043583149055\\
5	14.5326506621581\\
6	17.1809740364021\\
7	20.0358275161012\\
8	22.5650544425348\\
9	25.1280396806548\\
10	27.3294487177382\\
11	29.7279388753054\\
12	31.9432827057434\\
13	34.2645694267736\\
14	36.3078980913001\\
15	38.0446298310036\\
16	39.5593689387384\\
17	41.1592653399141\\
18	42.3288754370975\\
19	43.3829401871138\\
20	44.3447628218307\\
21	45.3613528979312\\
22	46.2862115848541\\
23	46.8839824972519\\
24	47.7246576008141\\
25	48.3687929830624\\
26	48.9804144308449\\
27	49.4558511383643\\
28	50.0223862875406\\
29	50.3382515573767\\
30	50.8142049430412\\
31	51.1068728704239\\
32	51.4129884142995\\
33	51.5700784332399\\
34	51.7371304673028\\
35	51.9139593898299\\
36	52.0074718858427\\
37	52.0722828449114\\
38	52.1108174479888\\
39	52.1322563317164\\
40	52.1681364037743\\
41	52.1742507740774\\
42	52.1693294299892\\
43	52.1688977616659\\
44	52.1688977616659\\
45	52.1688977616659\\
46	52.1688977616659\\
47	52.1688977616659\\
48	52.1688977616659\\
49	52.1688977616659\\
49	34.1693023540236\\
48	34.1693023540236\\
47	34.1693023540236\\
46	34.1693023540236\\
45	34.1693023540236\\
44	34.1693023540236\\
43	34.1693023540236\\
42	34.1694262662741\\
41	34.1777508776545\\
40	34.1711861017135\\
39	34.1594825975109\\
38	34.104768099393\\
37	34.0577686592914\\
36	34.0185384337867\\
35	33.9228368776827\\
34	33.8234193774794\\
33	33.7721258307008\\
32	33.5465487161976\\
31	33.3260205340514\\
30	32.9818512458462\\
29	32.7974032764656\\
28	32.682106453654\\
27	32.3522702467809\\
26	32.1529965598892\\
25	31.6478836321751\\
24	31.4670295913361\\
23	30.8610155096005\\
22	30.2897447430032\\
21	29.6190511505558\\
20	28.9031796838664\\
19	28.2297197118996\\
18	27.5455623109637\\
17	26.4946527634377\\
16	25.4298047064668\\
15	24.3817083863905\\
14	23.0887020693811\\
13	21.8919203399562\\
12	20.5822375315288\\
11	19.4129075766331\\
10	17.9734560531882\\
9	16.4498805037412\\
8	14.8371921723408\\
7	12.9252701702739\\
6	11.3199043335402\\
5	9.30252921945164\\
4	7.37122402582956\\
3	5.44296746690072\\
2	3.11537395011982\\
1	1.31411661975003\\
}--cycle;
\addplot [color=myblue, line width=2.0pt]
  table[row sep=crcr]{%
  0   0\\
1	2.50193242476186\\
2	4.95631920106466\\
3	7.41419978974564\\
4	9.63779117036751\\
5	11.9175899408049\\
6	14.2504391849712\\
7	16.4805488431875\\
8	18.7011233074378\\
9	20.788960092198\\
10	22.6514523854632\\
11	24.5704232259693\\
12	26.2627601186361\\
13	28.0782448833649\\
14	29.6983000803406\\
15	31.213169108697\\
16	32.4945868226026\\
17	33.8269590516759\\
18	34.9372188740306\\
19	35.8063299495067\\
20	36.6239712528485\\
21	37.4902020242435\\
22	38.2879781639286\\
23	38.8724990034262\\
24	39.5958435960751\\
25	40.0083383076187\\
26	40.5667054953671\\
27	40.9040606925726\\
28	41.3522463705973\\
29	41.5678274169211\\
30	41.8980280944437\\
31	42.2164467022377\\
32	42.4797685652485\\
33	42.6711021319703\\
34	42.7802749223911\\
35	42.9183981337563\\
36	43.0130051598147\\
37	43.0650257521014\\
38	43.1077927736909\\
39	43.1458694646136\\
40	43.1696612527439\\
41	43.1760008258659\\
42	43.1693778481317\\
43	43.1691000578448\\
44	43.1691000578448\\
45	43.1691000578448\\
46	43.1691000578448\\
47	43.1691000578448\\
48	43.1691000578448\\
49	43.1691000578448\\
};
\addlegendentry{$\lambda\text{=2}$}

\addplot[area legend, draw=mygreen, fill=mygreen, fill opacity=0.2, forget plot]
table[row sep=crcr] {%
x	y\\
0   0\\
1	3.3926275682848\\
2	5.75930091314732\\
3	8.25439879532083\\
4	10.5441311230055\\
5	12.5336364861145\\
6	14.2607761768403\\
7	16.2423824263045\\
8	17.7337115339871\\
9	19.4394845763661\\
10	20.9845373855114\\
11	22.0792138227822\\
12	23.4879634560568\\
13	24.9642819968042\\
14	26.3361326147128\\
15	27.4269453676147\\
16	28.1365253854545\\
17	29.0906690104725\\
18	30.3820751817646\\
19	31.3727550787652\\
20	31.9049114876835\\
21	32.7225960552991\\
22	33.3481134533472\\
23	34.0422140676238\\
24	34.7467869772092\\
25	35.2482174844241\\
26	35.6938403265812\\
27	36.1206325668433\\
28	36.4869511982949\\
29	36.8291777531083\\
30	37.1071004037322\\
31	37.28929793226\\
32	37.3200930429127\\
33	37.4441942319657\\
34	37.540649963159\\
35	37.6089424789959\\
36	37.677048532227\\
37	37.7018796056384\\
38	37.7747948008681\\
39	37.7817978047797\\
40	37.7975382857611\\
41	37.8043532686747\\
42	37.7959305119219\\
43	37.7960903075025\\
44	37.7960907142633\\
45	37.7960907142633\\
46	37.7960907142633\\
47	37.7960907142633\\
48	37.7960907142633\\
49	37.7960907142633\\
49	21.9745743453596\\
48	21.9745743453596\\
47	21.9745743453596\\
46	21.9745743453596\\
45	21.9745743453596\\
44	21.9745743453596\\
43	21.9745742418702\\
42	21.9745335065502\\
41	21.9739095702949\\
40	21.9738813944456\\
39	21.9771528227097\\
38	21.9527938411261\\
37	21.9174058923014\\
36	21.8743667961828\\
35	21.9193103200879\\
34	21.8526092959341\\
33	21.820414350796\\
32	21.7576781414427\\
31	21.6839991401883\\
30	21.5353969313293\\
29	21.3440804351115\\
28	21.3151011990981\\
27	21.2081338899867\\
26	20.952574308894\\
25	20.6315137104683\\
24	20.4320016726946\\
23	20.1398267183692\\
22	19.8188846002473\\
21	19.3552201294087\\
20	18.8488788690491\\
19	18.2098382292539\\
18	17.7042223657973\\
17	17.215667740551\\
16	16.4770131881503\\
15	15.8660054579197\\
14	15.1625275082077\\
13	13.9811277316788\\
12	12.8383705687737\\
11	12.2093786399307\\
10	11.0354175748587\\
9	10.0027812006316\\
8	8.6765059546501\\
7	7.47568116585045\\
6	6.28508924009862\\
5	5.38808050998785\\
4	3.97459329537258\\
3	2.74294018689862\\
2	1.58616770184121\\
1	0.446107440668074\\
0   0\\
}--cycle;
\addplot [color=mygreen, line width=2.0pt]
  table[row sep=crcr]{%
  0   0\\
1	1.91936750447644\\
2	3.67273430749426\\
3	5.49866949110972\\
4	7.25936220918905\\
5	8.96085849805115\\
6	10.2729327084695\\
7	11.8590317960775\\
8	13.2051087443186\\
9	14.7211328884988\\
10	16.0099774801851\\
11	17.1442962313565\\
12	18.1631670124152\\
13	19.4727048642415\\
14	20.7493300614602\\
15	21.6464754127672\\
16	22.3067692868024\\
17	23.1531683755117\\
18	24.0431487737809\\
19	24.7912966540096\\
20	25.3768951783663\\
21	26.0389080923539\\
22	26.5834990267972\\
23	27.0910203929965\\
24	27.5893943249519\\
25	27.9398655974462\\
26	28.3232073177376\\
27	28.664383228415\\
28	28.9010261986965\\
29	29.0866290941099\\
30	29.3212486675307\\
31	29.4866485362242\\
32	29.5388855921777\\
33	29.6323042913808\\
34	29.6966296295465\\
35	29.7641263995419\\
36	29.7757076642049\\
37	29.8096427489699\\
38	29.8637943209971\\
39	29.8794753137447\\
40	29.8857098401033\\
41	29.8891314194848\\
42	29.885232009236\\
43	29.8853322746864\\
44	29.8853325298114\\
45	29.8853325298114\\
46	29.8853325298114\\
47	29.8853325298114\\
48	29.8853325298114\\
49	29.8853325298114\\
};
\addlegendentry{$\lambda\text{=3}$}

\end{axis}

\end{tikzpicture}%